\documentclass[12pt]{article}
\usepackage{mathrsfs}
\usepackage{graphicx}
\usepackage{enumerate}
\usepackage{amsfonts,amsmath,amssymb}
\usepackage{geometry}
\usepackage{float} 

\newcommand{\n}{\noindent}

\def\XXint#1#2#3{{\setbox0=\hbox{$#1{#2#3}{\int}$}
		\vcenter{\hbox{$#2#3$}}\kern-.5\wd0}}

\newtheorem{oldtheorem}{Theorem}[section]
\newtheorem{oldassertion}[oldtheorem]{Assertion}
\newtheorem{oldproposition}[oldtheorem]{Proposition}
\newtheorem{oldremark}[oldtheorem]{Remark}
\newtheorem{oldlemma}[oldtheorem]{Lemma}
\newtheorem{olddefinition}[oldtheorem]{Definition}
\newtheorem{oldclaim}[oldtheorem]{Claim}
\newtheorem{oldcorollary}[oldtheorem]{Corollary}

\newenvironment{theorem}{\begin{oldtheorem}$\!\!\!${\bf.}}{\end{oldtheorem}}

\newenvironment{lemma}{\begin{oldlemma}$\!\!\!${\bf.}}{\end{oldlemma}}

\newbox\qedbox

\newenvironment{proof}{\smallskip\noindent{\bf Proof.}\hskip \labelsep}%
{\hfill\penalty10000\copy\qedbox\par\medskip}

\newenvironment{acknowledgements}{\smallskip\noindent{\bf Acknowledgements.}
	\hskip\labelsep}{}

\providecommand{\keywords}[1]{\textbf{\textbf{Keywords---}} #1}

\title{Existence of Ground State and Excited Spinning $Q$-Vortex Solitons on Finite Domains}

\author{
	Caroline Brumelot\\Department of Industrial Engineering and Operations Research\\ Columbia University\\New York, NY 10027, USA\\
	\\
	Luciano Medina\footnote{lmedina2@albany.edu}\\Department of Mathematics and Statistics\\ Computational Physics Program, Department of Physics\\ State University of New York at Albany\\ Albany, NY 12222, USA}\author{
    Caroline Brumelot \\
    Department of Industrial Engineering and Operations Research \\
    Columbia University \\
    New York, NY 10027, USA
    \and
    Luciano Medina\footnote{Email address: lmedina2@albany.edu} \\
    Department of Mathematics and Statistics \\
    Computational Physics Program, Department of Physics \\
    State University of New York at Albany \\
    Albany, NY 12222, USA
}

\date{}

\begin{document}
	\maketitle
\begin{abstract}
    We establish the existence of spinning $Q$-vortex solitons in a complex scalar field theory with a sextic potential on a finite domain. By reducing the governing equation to a nonlinear boundary value problem, we use variational methods to prove the existence of at least two distinct types of solutions: a ground state solution obtained via constrained minimization and an excited state of the saddle-point type obtained via the Mountain Pass Theorem. We derive bounds for the angular frequency $\omega$, the wave amplitude, and the domain size $P$, and provide explicit estimates for the exponential decay of the solutions. Furthermore, we implement a spectral-Galerkin formulation to numerically compute the profiles of fundamental $Q$-vortices, illustrating the saturation behavior of the soliton's amplitude and the asymptotic dependence of the frequency on a prescribed reduced norm and vortex winding number, as well as verifying the theoretical results and visualizing the topological phase structure of the solutions.
\end{abstract}
	
	\n {\bf 2020 Mathematics Subject Classification.} 35J20, 35Q55, 35Q60, 65N30.
	
	\medskip
	\keywords{Spinning $Q$-vortices, Optical Solitons, Variational Methods, Constrained Minimization, Mountain Pass Theorem, Sextic Potential, Finite Domain}

	\section{Introduction}
	\setcounter{equation}{0}

In the context of field theoretic models of high energy physics, solitons are localized finite energy solutions of nonlinear field equations and realizations of elementary particles. Solitons, in most nonlinear field theories, possess strong regularity properties but typically lack a nonzero angular momentum, which is a desired characteristic of elementary particles \cite{Derrick1964, Mandal2022, VW}. There are two classes of solitons: topological, having a conserved topological charge, and nontopological, having a conserved Noether charge \cite{Coleman, Hamada2024, KKK, Tai2024}. Examples of topological and nontopological solitons include monopoles, vortices, domain walls, soliton stars, boson stars, hadron structures, and of our particular interest, $Q$-balls \cite{Ansari2024, Birse, BHZ, BCHL, DW, Kusenko, LeePang, Ivashkin2025}.

From a mathematical perspective, $Q$-balls are examples of the class of hylomorphic solitons and their existence has been well studied by Benci and others \cite{Benci3,Benci4,Benci5}. While the work of Volkov and W\"ohnert \cite{VW} presented the first numerical evidence of static spinning $Q$-balls, the field has seen significant recent activity. Notably, Almumin et al. \cite{Almumin2024} recently investigated slowly rotating $Q$-balls, confirming the existence of metastable states with small angular momentum. Furthermore, studies by Zhang et al. \cite{Zhang2025} have explored superradiance in rotating solitons, linking the ``spinning" nature of these solutions to wider stability phenomena. Beyond the global $U(1)$ symmetry, significant progress has also been made in the study of gauged or electrically charged $Q$-balls. Recent works have established the existence of such charged solutions in multi-scalar theories \cite{Brihaye2024}, proven their existence in coupled Abelian-Higgs models \cite{SuHan2025}, and explored their complex dynamical interactions, where electromagnetic repulsion plays a critical role \cite{GaugedCollisions2024}.

Our aim is to establish a rigorous mathematical existence theory for spinning $Q$-vortices-generalizations of $Q$-balls in $2+1$ dimensions analyzed numerically by Volkov and W\"ohnert \cite{VW} and to use our theory to study their profiles. To this end, following the formulation in \cite{VW}, consider the theory of a complex scalar field in 3+1 dimensional flat Minkowski spacetime defined by the Lagrangian density
\begin{align}
    \mathcal{L}=\partial_{\mu}\Phi\partial^{\mu}\Phi^*-U(|\Phi|).\label{theory}
\end{align}
The global $U(1)$ symmetry of the Lagrangian under $\Phi \to e^{i\alpha}\Phi$ leads to the conserved Noether charge
\begin{align}
    Q = \frac{1}{i} \int (\Phi^* \dot{\Phi} - \Phi \dot{\Phi}^*) d^3x. \label{NoetherCharge}
\end{align}
The variation of the action with respect to the field $\Phi^*$ yields the general Euler-Lagrange equation of motion:
\begin{align}
    \partial_{\mu}\partial^{\mu}\Phi + \frac{\partial U}{\partial \Phi^*} = 0. \label{fieldEq}
\end{align}
Assuming cylindrical symmetry where the field is independent of $z$, and employing polar coordinates $(t,\rho,z,\theta)$, we seek spinning $Q$-vortices via the ansatz
\begin{align}
    \Phi=\phi(\rho)e^{i\omega t+iN\theta},\label{ansazts}
\end{align}
where $\omega\in\mathbb{R}$ is the wave frequency and $N\in\mathbb{Z}$ is the rotational quantum number or vortex number.

The energy per unit vortex length is
\begin{align}
    E=2\pi\int_0^{\infty}\rho\left(\omega^2\phi^2+\phi_{\rho}^2+\dfrac{N^2}{\rho^2}\phi^2+U(\phi)\right)d\rho,\label{energy1}
\end{align}
and the angular momentum per unit length is given by
\begin{align}
    J=4\pi\omega N\int_0^{\infty}\rho\phi^2d\rho\equiv NQ.
\end{align}
When $N\neq 0$, the vortex rotates around the $z$-axis with angular momentum $J$.

From \eqref{energy1}, we see that finite energy solutions on the unbounded domain require $\phi(0)=0$ and decay $\phi(\rho)\rightarrow 0$ as $\rho\rightarrow\infty$. To rigorously establish existence and facilitate numerical computation, we study the problem on a finite domain $D_P = \{ x \in \mathbb{R}^2 : |x| < P \}$, imposing the Dirichlet boundary condition $\phi(P)=0$ for some large $P>0$. This formulation serves dual purposes: it models physical confinement and acts as a computational proxy for the free-space solution, consistent with the numerical approach of Volkov and W\"ohnert \cite{VW}. We establish a priori bounds and decay estimates in Section 2 to quantify the behavior of these solutions as $\rho \to P$.

Under the explicit potential
\begin{align}
    U(\phi)=\lambda\left(\phi^6-a\phi^4+b\phi^2\right), \qquad \lambda,a,b>0,\qquad b>a^2/4,
\end{align}
the existence of $Q$-vortices reduces to the two-point boundary value problem:
\begin{align}
    \left\{\begin{array}{cc}
        &\phi_{\rho\rho}+\dfrac{1}{\rho}\phi_{\rho}-\dfrac{N^2}{\rho^2}\phi-\lambda(6\phi^5-4a\phi^3+2b\phi)+\omega^2\phi=0, \quad\rho\in(0,P),\label{DE}\\
        &\phi(0)=0=\phi(P).
    \end{array}\right.
\end{align}

In our mathematical analysis, it is convenient to distinguish between the reduced norm functional $\tilde{Q}(\phi)$, defined by:
\begin{align}
    \tilde{Q}(\phi):=4\pi\int_0^P\rho\phi^2d\rho,\label{ConstraintFunctional1}
\end{align}
and the constraint equation $\tilde{Q}(\phi)=\tilde{Q}_0$. We use the term prescribed norm for the constant $\tilde{Q}_0$, which serves as the constraint value in our minimization problem. In the context of nonlinear optics, this quantity is often referred to as the beam power, while in particle physics it represents the static geometric contribution to the Noether charge. This distinction aligns with recent rigorous studies on the existence of optical vortices \cite{ChenSu2022, Guo2019, Greco2016, Medina2017, Medina2021, Medina2023}, where the beam power is prescribed rather than the frequency.

We now state our main analytical results governing solutions to \eqref{DE}, followed by a summary of our numerical findings.

\begin{theorem}
    Let $\phi\in\mathcal{C}[0,P]\cap\mathcal{C}^2(0,P)$ be a nontrivial classical solution of \eqref{DE}. A necessary condition for existence is
    \begin{equation}
        2\lambda\left(b - \frac{a^2}{3}\right) + \frac{N^2}{ P^2} < \omega^2.
        \label{OmegaLowerBnd}\end{equation}
    Furthermore, if $\omega^2 < 2\lambda b + \frac{N^2}{P^2}$, the solution is uniformly bounded by
    \begin{equation}
        \phi^2(\rho) < \frac{2a}{3} \quad \text{for all } \rho \in [0, P].
    \end{equation}
    Additionally, under the same condition on $\omega$, the solution decays exponentially fast near the boundary:
    \begin{equation}
        \phi^2(\rho)\leq \frac{2a}{3}\exp(-\sigma(\rho-P_0))\qquad\text{for every $\rho\in[P_0,P]$,}
    \end{equation}
    with $P_0\in(0,P)$ sufficiently large and decay rate $\sigma = \sqrt{\frac{N^2}{P^2}+2\lambda b-\omega^2}$.
\end{theorem}

\begin{theorem}
    Let $|N|\geq 1$ and $\lambda,a,b>0$ such that $b>a^2/4$. For any $\omega$ satisfying
    \begin{equation}
        2\lambda\left(b-\dfrac{1}{4}a^2\right)<\omega^2<2\lambda b
    \end{equation}
    and $P$ sufficiently large (specifically $P > P^*$ as defined in Lemma 3.2), there exist at least two different types of positive solutions, $\phi(\rho)>0$ for every $\rho\in(0,P)$, satisfying \eqref{DE}. One solution is a minimum and the other a saddle point of the corresponding action functional. There also exist infinitely many pairs of solutions (not necessarily positive) to \eqref{DE}.
\end{theorem}

\begin{theorem}
    Let $|N|\geq 1$ and $\lambda,a,b>0$ such that $b>a^2/4$. There exists a solution pair $(\phi_0,\omega_0)$ to \eqref{DE} such that $\phi_0(\rho)>0$ for all $\rho\in(0,P)$, where $\omega_0$ appears as the Lagrange multiplier of a constrained minimization problem. Additionally, if $\omega_0^2 < 2\lambda b$, the prescribed norm $\tilde{Q}_0$ must satisfy the lower bound:
    \begin{align}
        \tilde{Q}_0 > \frac{\pi |N|}{a\lambda}.
    \end{align}
\end{theorem}

\medskip
\noindent\textbf{Numerical Implications.} We complement our existence theory with a spectral-Galerkin framework and validate the theoretical bounds and characterize the physical structure of the vortices, revealing three distinct behaviors, in agreement with the numerical work of Volkov and W\"ohnert \cite{VW}:
\begin{enumerate}
    \item \textit{Amplitude Saturation:} Consistent with Theorem 1.1, the field amplitude remains strictly bounded. As the prescribed norm $\tilde{Q}_0$ increases, the soliton profile transitions from a Gaussian-like shape to a saturated ``flat-top" structure, asymptotically approaching the theoretical ceiling $\phi_{max} \approx \sqrt{2a/3}$.
    \item \textit{Nonlinear Dispersion:} The eigenfrequency $\omega^2$ exhibits a monotonic, nonlinear dependence on the prescribed norm. Our computations map the full frequency existence range, showing the frequency descending from the linear limit $\omega_{max}^2 = 2\lambda b$ toward the critical existence threshold $\omega_{min}^2=2\lambda\left(b-\frac{1}{4}a^2\right)$ as the vortex grows.
    \item \textit{Topological Geometry:} For higher winding numbers ($|N| \ge 2$), the numerical profiles demonstrate the dominant role of the centrifugal barrier ($N^2/\rho^2$). Increasing the topological charge significantly widens the vortex core, displacing the energy density radially outward and confirming the existence of topologically distinct states with phase winding numbers corresponding to $N$.
\end{enumerate}

	\section{Proof of Theorem 1.1}
	\setcounter{equation}{0}
	
	In this section, we establish Theorem 1.1 as a result of the three separate lemmas below. We first establish necessary bounds on the frequency $\omega$ and the amplitude $\phi$, and then use these bounds to prove the exponential decay of the solution near the boundary.
	
	\begin{lemma}
		If $\phi\in\mathcal{C}[0,P]\cap\mathcal{C}^2(0,P)$ is a nontrivial solution of \eqref{DE}, then
		\begin{align}
			2\lambda\left(b-\dfrac{1}{3}a^2\right)+\dfrac{N^2}{P^2}<\omega^2.
		\end{align}
	\end{lemma}
	\begin{proof}
		If $\phi$ is a solution to \eqref{DE}, then $-\phi$ is also a solution. Hence, without loss of generality, there exists a point $\rho_0\in(0,P)$ such that $\phi(\rho_0)>0$ is a local maximum, so $\phi_{\rho}(\rho_0)=0$ and $\phi_{\rho\rho}(\rho_0)\leq 0$. Substituting these into \eqref{DE}, we get
		\begin{align}
			0 \geq \phi_{\rho\rho}(\rho_0) &= \left(\dfrac{N^2}{\rho_0^2}+\lambda(6\phi^4(\rho_0)-4a\phi^2(\rho_0)+2b)-\omega^2\right)\phi(\rho_0).
		\end{align}
		Since $\phi(\rho_0) > 0$, the term in the parenthesis must be non-positive:
		\begin{align}
			\omega^2 &\geq \dfrac{N^2}{\rho_0^2} + 2\lambda(3\phi^4(\rho_0)-2a\phi^2(\rho_0)+b).
		\end{align}
		The quadratic polynomial $f(u) = 3u^2 - 2au + b$ (with $u=\phi^2$) achieves its global minimum at $u = a/3$. The minimum value is $f(a/3) = b - a^2/3$. Using this and the fact that $\rho_0 < P$ (so $1/\rho_0^2 > 1/P^2$), we obtain the strict inequality:
		\begin{align}
			\omega^2 &> \dfrac{N^2}{P^2} + 2\lambda\left(b-\dfrac{1}{3}a^2\right).
		\end{align}
		$\qquad\square$
	\end{proof}
	
	\begin{lemma}
		If $\phi\in\mathcal{C}^0[0,P]\cap\mathcal{C}^2(0,P)$ is a solution of \eqref{DE} and $\omega^2<2\lambda b+\dfrac{N^2}{P^2}$, then $\phi^2(\rho)<\dfrac{2a}{3}$ for all $\rho\in[0,P]$.
	\end{lemma}
	\begin{proof}
		Let $\phi\in\mathcal{C}^0[0,P]\cap\mathcal{C}^2(0,P)$ be a solution of \eqref{DE} with $\omega^2<2\lambda b+\dfrac{N^2}{P^2}$. As before, assume $\phi$ has a positive global maximum at $\rho_0\in(0,P)$ such that $\phi_{\rho}(\rho_0)=0$ and $\phi_{\rho\rho}(\rho_0)\leq 0$. Substituting into \eqref{DE}:
		\begin{align}
			0 \geq \phi_{\rho\rho}(\rho_0)&=\left(\dfrac{N^2}{\rho_0^2}+2\lambda b - \omega^2\right)\phi(\rho_0)+6\lambda\phi^3(\rho_0)\left(\phi^2(\rho_0)-\dfrac{2a}{3}\right).
		\end{align}
		By the hypothesis on $\omega^2$, we have $\dfrac{N^2}{\rho_0^2}+2\lambda b - \omega^2 > \dfrac{N^2}{P^2}+2\lambda b - \omega^2 > 0$. Since $\phi(\rho_0) > 0$, the first term on the right-hand side is strictly positive. For the sum to be non-positive, the second term must be negative:
		\begin{align}
			6\lambda\phi^3(\rho_0)\left(\phi^2(\rho_0)-\dfrac{2a}{3}\right) < 0.
		\end{align}
		This implies $\phi^2(\rho_0)<\dfrac{2a}{3}$. Since $\phi(\rho_0)$ is the global maximum, it follows that $\phi^2(\rho) \leq \phi^2(\rho_0) < \dfrac{2a}{3}$ for all $\rho\in[0,P]$.
		$\qquad\square$
	\end{proof}
	
	\begin{lemma}
		Let $\phi\in\mathcal{C}^0[0,P]\cap\mathcal{C}^2(0,P)$ is a solution of \eqref{DE}. If $\omega^2$ satisfies the inequality
		\begin{equation}
			\omega^2<2\lambda b+\dfrac{N^2}{P^2},
		\end{equation}
		then there exist constants $P_0\in(0,P)$ and $\sigma>0$ (depending only on the system parameters), such that
		\begin{equation*}
			\phi^2(\rho)\leq \frac{2a}{3}\exp(-\sigma(\rho-P_0))\qquad\text{for every $\rho\in[P_0,P]$.}
		\end{equation*}
	\end{lemma}
	\begin{proof}
		Rewrite \eqref{DE} as:
		\begin{align}
			\Delta\phi = \left(\dfrac{N^2}{\rho^2}-\omega^2+\lambda\left(6\phi^4-4a\phi^2+2b\right)\right)\phi.
		\end{align}
		We compute the Laplacian of $w = \phi^2$:
		\begin{align}
			\Delta w = 2\phi\Delta\phi + 2|\nabla\phi|^2 \geq 2\phi\Delta\phi = 2\left(\dfrac{N^2}{\rho^2}-\omega^2+2\lambda b + \lambda(6\phi^4-4a\phi^2)\right)w.
		\end{align}
		Since $\phi(P)=0$ and $\phi$ is continuous, for any $\epsilon > 0$, there exists a $P_0$ sufficiently close to $P$ such that the nonlinear term satisfies $6\lambda\phi^4 - 4a\lambda\phi^2 > -\epsilon$ for all $\rho \in [P_0, P]$.
		
		Let $2\delta = \frac{N^2}{P^2} + 2\lambda b - \omega^2$, which is strictly positive by hypothesis. We choose $P_0$ such that the nonlinear term is small enough, specifically choosing $\epsilon = \delta$. Then for all $\rho \in [P_0, P]$:
		\begin{align}
			\Delta w \geq 2\left( \frac{N^2}{P^2} - \omega^2 + 2\lambda b - \delta \right)w = 2\delta w.
		\end{align}
		Now, define the comparison function $\xi(\rho) = \frac{2a}{3} e^{-\sigma(\rho-P_0)}$, where we set $\sigma = \sqrt{2\delta}$. Computing the Laplacian of $\xi$:
		\begin{align}
			\Delta \xi = \xi'' + \frac{1}{\rho}\xi' = \sigma^2 \xi - \frac{\sigma}{\rho}\xi < \sigma^2 \xi = 2\delta \xi.
		\end{align}
		Define $v(\rho) = w(\rho) - \xi(\rho)$. Then:
		\begin{align}
			\Delta v = \Delta w - \Delta \xi > 2\delta w - 2\delta \xi = 2\delta v.
		\end{align}
		Thus $(\Delta - 2\delta)v > 0$ on $(P_0, P)$. By the maximum principle, the maximum of $v$ must occur on the boundary.
		At the inner boundary $\rho = P_0$: $v(P_0) = \phi^2(P_0) - \frac{2a}{3}$. By Lemma 2.2, we know $\phi^2(\rho) < \frac{2a}{3}$ globally, so $v(P_0) < 0$.
		At the outer boundary $\rho = P$: $v(P) = 0 - \frac{2a}{3} e^{-\sigma(P-P_0)} < 0$.
		Since $v < 0$ on the boundary, $v(\rho) \leq 0$ for all $\rho \in [P_0, P]$.
		Therefore, $\phi^2(\rho) \leq \frac{2a}{3} e^{-\sqrt{2\delta}(\rho-P_0)}$.
		$\qquad\square$
	\end{proof}

	\section{Proof of Theorem 1.2}
	\setcounter{equation}{0}
	Fundamental $Q$-vortices are nontrivial positive solutions of \eqref{DE}. We first use a direct minimization approach to establish the existence of fundamental $Q$-vortices which arise as minima of a suitable action functional. We then utilize the Mountain Pass Theorem to establish the existence of a second, distinct solution of the saddle point type.
	
	To this end, let the action functional $I_{\omega}:H\rightarrow\mathbb{R}$ be given by
	\begin{align}
		I_{\omega}(\phi)=\int_0^{P}\left\{\dfrac{1}{2}\left(\rho\phi_{\rho}^2+\dfrac{N^2}{\rho}\phi^2-\omega^2\rho\phi^2\right)+\lambda\rho(\phi^6-a\phi^4+b\phi^2)\right\}d\rho,\label{Daction}
	\end{align}
	where $H$ is the completion of $X=\left\{\phi\in\mathcal{C}^1[0,P]: \phi(0)=0=\phi(P)\right\}$ endowed with the inner product
	\begin{align}
		(\phi,\tilde{\phi})&=\int_0^P\left\{\rho \phi_{\rho}\tilde{\phi}_{\rho}+\dfrac{1}{\rho}\phi\tilde{\phi} \right\}d\rho,\quad \phi,\tilde{\phi}\in H,\label{innerP}
	\end{align}
	and norm $||\phi||^2=(\phi,\phi)$.
	
	Consider the direct minimization problem
	\begin{align}
		\eta_0=\inf\{I_{\omega}(\phi):\phi\in H\}.\label{Dmin}
	\end{align}
	Rewrite the functional \eqref{Daction} in the form
	\begin{align}
		I_{\omega}(\phi)=\int_0^{P}\left\{\dfrac{1}{2}\left(\rho\phi_{\rho}^2+\dfrac{N^2}{\rho}\phi^2\right)+\lambda\rho h(\phi)\right\}d\rho
	\end{align}
	with
	\begin{align}
		h(\phi)=\phi^6-a\phi^4+\left(b-\dfrac{\omega^2}{2\lambda}\right)\phi^2.
	\end{align}
	Since the coefficient of $\phi^6$ is positive, $h(\phi)$ is bounded from below by some constant $h_{\min}$. It follows that
	\begin{align}
		I_{\omega}(\phi)\geq\dfrac{1}{2}||\phi||^2+\dfrac{1}{2}\lambda h_{\min}P^2\label{coercive}
	\end{align}
	and $I_{\omega}$ is coercive on $H$. Thus the direct minimization problem \eqref{Dmin} is well-defined.
	
	The following lemma proves that the functional $I_{\omega}$ satisfies the Palais-Smale (PS) condition.
	
	\begin{lemma}
		Let $\{\phi_j\}_{j=1}^{\infty}$ be a sequence in $H$ such that $|I_{\omega}(\phi_j)|\leq \alpha$ for some constant $\alpha>0$ and $I'_{\omega}(\phi_j)\rightarrow 0$ as $j\rightarrow\infty$. Then there exists a subsequence of $\{\phi_j\}_{j=1}^{\infty}$ which converges strongly to some element $\phi$ in $H$.
	\end{lemma}
	\begin{proof}
		Suppose that $\{\phi_j\}_{j=1}^{\infty}$ is a sequence in $H$ such that $|I_{\omega}(\phi_j)|\leq \alpha$ for some constant $\alpha>0$. Then the coercive bound \eqref{coercive} gives us the boundedness of the sequence $\{\phi_j\}_{j=1}^{\infty}$ in $H$. We may treat each $\phi_j$, $j=1,2,\ldots$, as a real-valued radially symmetric function defined over the disc $D_P=\{(x,y)\in\mathbb{R}^2:x^2+y^2\leq P^2\}$ and vanishing on the boundary $\partial D_P$. From the inequality
		\begin{align}
			\int_0^P\rho\phi^2d\rho\leq P^2\int_0^P\dfrac{\phi^2}{\rho}d\rho,
		\end{align}
		it follows that $H$ is an embedded subspace of the standard Sobolev space $W^{1,2}_0(D_P)$, composed of radially symmetric functions such that each element $\phi$ of $H$ satisfies $\phi(0)=0$. Hence, $\{\phi_j\}_{j=1}^{\infty}$ is bounded in $W^{1,2}_0(D_P)$. Without loss of generality, we may suppose that $\phi_j\rightharpoonup \phi$ weakly in $W^{1,2}_0(D_P)$ as $j\rightarrow\infty$.
		The compact embedding $W^{1,2}(D_P)\subset\subset L^p(D_P)$ $(p\geq 1)$ then gives the strong convergence of $\phi_j\rightarrow\phi$ in $L^p(D_P)$ as $j\rightarrow\infty$ for every $p\geq 1$. Moreover, $\phi$ is radially symmetric and we may write $\phi=\phi(\rho)$, which satisfies $\phi(P)=0$.
		
		By hypothesis, we also suppose that $I'_{\omega}(\phi_j)\rightarrow 0$ as $j\rightarrow\infty$ or, equivalently, that there is a sequence $\epsilon_j\geq 0$ for every $j=1,2,\ldots,$ such that
		\begin{equation}
			|(I'_{\omega}(\phi_j),\tilde{\phi})|\leq \epsilon_j||\tilde{\phi}||\label{PS2}
		\end{equation}
		for every $\tilde{\phi}\in H$ and $\epsilon_j\rightarrow 0$ as $j\rightarrow\infty$. We also use the notation $(\cdot,\cdot)$ to denote the duality pairing between $H$ and its dual space $H^{-1}$.
		
		Letting $j\rightarrow\infty$ in \eqref{PS2}, we get
		\begin{align}
			0=\int_0^P\left\{\rho\phi_{\rho}\tilde{\phi}_{\rho}+\dfrac{N^2}{\rho}\phi\tilde{\phi}+\lambda\rho\left(6\phi^5\tilde{\phi}-4a\phi^3\tilde{\phi}+2\left(b-\dfrac{\omega^2}{2\lambda}\right)\phi\tilde{\phi}\right)\right\}d\rho \label{PSP2}
		\end{align}
		for every $\tilde{\phi}\in H$. Let $\tilde{\phi}=\phi_j-\phi$ in \eqref{PS2} and \eqref{PSP2}. Inserting the result of \eqref{PSP2} into the resulting \eqref{PS2}, we arrive at
		\begin{align}
			&\left|\int_0^P\left\{\rho(\phi_{j,\rho}-\phi_{\rho})^2+\dfrac{N^2}{\rho}(\phi_j-\phi)^2+\lambda\rho\left(6(\phi_j^5-\phi^5)-4a(\phi_j^3-\phi^3)\right.\right.\right.\nonumber\\
			&\left.\left.\left.+2\left(b-\dfrac{\omega^2}{2\lambda}\right)(\phi_j-\phi)\right)(\phi_j-\phi)\right\}d\rho\right|\leq\epsilon_j||\phi_j-\phi||.
		\end{align}
		Consequently, using the triangle inequality, we isolate the norm $||\phi_j - \phi||^2$:
		\begin{align}
			||\phi_j-\phi||^2&\leq\lambda\int_0^P\rho\left(6|\phi_j^5-\phi^5|+4a|\phi_j^3-\phi^3|+2\left|b-\dfrac{\omega^2}{2\lambda}\right||\phi_j-\phi|\right)|\phi_j-\phi|d\rho\nonumber\\
			&+\epsilon_j||\phi_j-\phi||.
		\end{align}
		To prove strong convergence, we show that the integral term on the right-hand side vanishes as $j \to \infty$. We analyze the highest order term $\int_0^P \rho |\phi_j^5 - \phi^5| |\phi_j - \phi| d\rho$. By the generalized H\"older inequality with exponents $p = 6/5$ and $q = 6$ (where $1/p + 1/q = 1$), we have:
		\begin{align}
			\int_0^P \rho |\phi_j^5 - \phi^5| |\phi_j - \phi| d\rho \leq \left( \int_0^P \rho |\phi_j^5 - \phi^5|^{6/5} d\rho \right)^{5/6} \left( \int_0^P \rho |\phi_j - \phi|^6 d\rho \right)^{1/6}.
		\end{align}
		Since $\phi_j \to \phi$ strongly in $L^6(D_P)$, the term $||\phi_j - \phi||_{L^6} \to 0$. Furthermore, the convergence $\phi_j \to \phi$ in $L^6$ implies $\phi_j^5 \to \phi^5$ strongly in $L^{6/5}(D_P)$ by the continuity of the Nemytskii operator. Thus, the product converges to zero. Similar arguments apply to the lower-order cubic and linear terms using appropriate H\"older exponents (e.g., $p=3/2, q=3$ or $p=2, q=2$).
		Therefore, the entire right-hand side tends to zero as $j \to \infty$, implying $||\phi_j - \phi|| \to 0$ strongly in $H$.
		$\qquad\square$
	\end{proof}
	
	Since $I_{\omega}$ is bounded below, coercive, and satisfies the PS condition, the global minimum exists. However, we must ensure this minimum is nontrivial (i.e., not the zero solution).
	
	\begin{lemma}
		Let $a,b,\lambda > 0$ such that $b > a^2/4$. For any $\omega$ satisfying
		\begin{align}
			2\lambda\left(b-\dfrac{1}{4}a^2\right)<\omega^2, \label{neccCond3}
		\end{align}
		there exists a $\phi_0\in H$ such that $I_{\omega}(\phi_0)<0$, provided $P > P^*$. The bound $P^*$ is given explicitly by
		\begin{align}
			P^* = \frac{\mathcal{B} + \sqrt{\mathcal{B}^2 + 4 \mathcal{A} \mathcal{C}}}{2\mathcal{A}},
		\end{align}
		where $\mathcal{A}, \mathcal{B}, \mathcal{C}$ are positive constants defined in the proof below.
	\end{lemma}
	
	\begin{proof}
		We construct a test function $\phi_0$ to verify the condition $I_{\omega}(\phi_0) < 0$. Let $t > 0$ be a parameter to be fixed later and define the trapezoidal function:
		\begin{align}
			\phi_0(\rho)=\left\{
			\begin{array}{cc}
				t\rho, &0\leq \rho < 1,\\
				t, &1\leq \rho \leq P-1,\\
				t(P-\rho), &P-1 < \rho\leq P.
			\end{array}\right.
		\end{align}
		First, we verify that $\phi_0 \in H$. The function $\phi_0$ is continuous on $[0, P]$ as the piecewise definitions agree at the transition points $\rho=1$ and $\rho=P-1$. Furthermore, $\phi_0(0)=0$ and $\phi_0(P)=0$, satisfying the boundary conditions. The weak derivative $\phi_{0,\rho}$ exists and is piecewise constant (taking values $t, 0, -t$), hence $\phi_{0,\rho} \in L^2(0, P)$. Thus, $\phi_0 \in H$.
		
		We now compute the action $I_{\omega}(\phi_0)$. The potential terms are dominated by the plateau region $[1, P-1]$. Integrating exactly, we find:
		\begin{align}
			\int_0^P \rho \phi_0^2 d\rho &= t^2\left(\frac{P^2}{2} - P + \frac{1}{3}\right), \\
			\int_0^P \rho \phi_0^4 d\rho &= t^4\left(\frac{P^2}{2} - P + \frac{1}{5}\right), \\
			\int_0^P \rho \phi_0^6 d\rho &= t^6\left(\frac{P^2}{2} - P + \frac{1}{7}\right).
		\end{align}
		The gradient term evaluates to $\int_0^P \rho \phi_{0,\rho}^2 d\rho = t^2 P$. For the centrifugal term, we use the upper bound valid for $P \ge 2$:
		\begin{align}
			\int_0^P \frac{\phi_0^2}{\rho} d\rho \le \int_0^1 \rho t^2 d\rho + \int_1^P \frac{t^2}{\rho} d\rho = \frac{t^2}{2} + t^2 \ln P.
		\end{align}
		Substituting these into $I_{\omega}(\phi_0)$ and grouping by powers of $P$, we obtain:
		\begin{align}
			I_{\omega}(\phi_0) \le -\mathcal{A} P^2 + \mathcal{B} P + \mathcal{C} \ln P,
		\end{align}
		where the coefficient of the quadratic term is:
		\begin{align}
			\mathcal{A} = -\frac{\lambda}{2} \left[ t^6 - a t^4 + \left(b - \frac{\omega^2}{2\lambda}\right)t^2 \right].
		\end{align}
		To ensure $\mathcal{A} > 0$, we choose $t$ to minimize the polynomial in brackets. Setting $t^2 = a/2$, we get:
		\begin{align}
			\mathcal{A} = -\frac{\lambda (a/2)}{2} \left( \frac{a^2}{4} - \frac{a^2}{2} + b - \frac{\omega^2}{2\lambda} \right) = \frac{\lambda a}{4} \left( \frac{\omega^2}{2\lambda} - \left(b - \frac{a^2}{4}\right) \right).
		\end{align}
		By hypothesis \eqref{neccCond3}, $\mathcal{A} > 0$. The linear coefficient $\mathcal{B}$ collects the remaining positive terms (gradient and potential corrections):
		\begin{align}
			\mathcal{B} = \frac{t^2}{2} + \lambda \left[ t^6 - at^4 + \left(b - \frac{\omega^2}{2\lambda}\right)t^2 \right] + \lambda \left( \frac{a}{5}t^4 - \frac{1}{7}t^6 - \frac{1}{3}\left(b - \frac{\omega^2}{2\lambda}\right)t^2 \right),
		\end{align}
		and $\mathcal{C} = \frac{N^2 t^2}{2}$. For large $P$, the term $\mathcal{C} \ln P$ is sub-dominant to $\mathcal{B} P$. We can strictly bound $\ln P < P$ for $P>1$ to simplify the sufficient condition to a quadratic inequality:
		\begin{align}
			-\mathcal{A} P^2 + (\mathcal{B} + \mathcal{C}) P < 0 \implies P > \frac{\mathcal{B} + \mathcal{C}}{\mathcal{A}}.
		\end{align}
		Thus, for sufficiently large $P$, $I_{\omega}(\phi_0) < 0$.
		$\qquad\square$
	\end{proof}
	
	From Lemma 3.2, we conclude that $\eta_0 = \inf I_\omega(\phi) < 0$. Since $I_\omega(0) = 0$, the minimizer $\tilde{\phi}$ must be nontrivial.
	
	We now establish the existence of a second, distinct solution using the Mountain Pass Theorem. First, we show that the functional possesses the requisite geometry.
	
	\begin{lemma}
		There is a $\delta >0$ and $C>0$ such that for any $\phi\in H$ with $0<||\phi||<\delta$, we have $I_{\omega}(\phi)>0$ and $I_{\omega}(\phi)\geq C$ for $||\phi||=\delta$.
	\end{lemma}
	\begin{proof}
		Rewrite the functional \eqref{Daction} in the form
		\begin{align}
			I_{\omega}(\phi)=\dfrac{1}{2}\int_0^{P}\left(\rho\phi_{\rho}^2+\dfrac{N^2}{\rho}\phi^2\right)d\rho+\lambda\int_0^P\rho g(\phi)d\rho
		\end{align}
		with
		\begin{align}
			g(\phi)=\phi^6-a\phi^4+\left(b-\dfrac{\omega^2}{2\lambda}\right)\phi^2.
		\end{align}
		For $\omega^2< 2\lambda b$, we have $\lim\limits_{\xi\rightarrow 0}\dfrac{g(\xi)}{\xi^2}=b-\dfrac{\omega^2}{2\lambda}> 0$.
		Thus, for any $\epsilon>0$ there is a $\xi_0>0$ such that
		\begin{equation}
			g(\xi)\geq \left(b-\dfrac{\omega^2}{2\lambda}-\epsilon\right)\xi^2\quad\text{for all $|\xi|<\xi_0$}.
		\end{equation}
		Let $\delta=\xi_0/\sqrt{2}>0$ and $\phi\in H$ such that $0<||\phi|| \le \delta$. Using the embedding $||\phi||_{\infty}\leq\sqrt{2}||\phi||$, we have $|\phi(\rho)| \le \xi_0$.
		Taking $\epsilon=\frac{1}{2}\left(b-\dfrac{\omega^2}{2\lambda}\right)$, we get
		\begin{align}
			I_{\omega}(\phi) \geq \dfrac{1}{2}||\phi||^2+\dfrac{1}{2}\lambda\left(b-\dfrac{\omega^2}{2\lambda}\right)\int_0^P\rho\phi^2d\rho \geq \dfrac{1}{2}||\phi||^2.
		\end{align}
		Thus, for $||\phi|| = \delta$, we have $I_{\omega}(\phi) \geq \frac{1}{2}\delta^2 := C > 0$.
		$\qquad\square$
	\end{proof}
	
	From Lemma 3.2, there exists $\phi_0 \in H$ such that $I_\omega(\phi_0) < 0$. We can choose $\delta$ small enough such that $||\phi_0|| > \delta$. Since $I_\omega(0)=0$, $I_\omega(\phi) \ge C > 0$ on the sphere $||\phi||=\delta$, and $I_\omega(\phi_0) < 0$, the functional exhibits the Mountain Pass geometry.
	
	Define the family of paths $\Gamma = \{ \gamma \in \mathcal{C}([0,1]; H) : \gamma(0)=0, \gamma(1)=\phi_0 \}$. The Mountain Pass level is defined as:
	\begin{align}
		c = \inf_{\gamma \in \Gamma} \max_{t \in [0,1]} I_\omega(\gamma(t)).
	\end{align}
	Since every path $\gamma$ from $0$ to $\phi_0$ must intersect the sphere $||\phi||=\delta$, we have:
	\begin{align}
		c \geq \inf_{||\phi||=\delta} I_\omega(\phi) \geq C > 0.
	\end{align}
	By the Mountain Pass Theorem \cite{Jabri, PR}, since $I_\omega$ satisfies the PS condition, there exists a critical point $\phi_{MP} \in H$ such that $I_\omega(\phi_{MP}) = c$ and $I'_\omega(\phi_{MP}) = 0$.
	
We conclude that we have found at least two distinct critical points for the functional $I_\omega$. To ensure these solutions are physically relevant (positive), we employ a truncation argument. Define the modified nonlinearity $h^+(\phi)$ by:
\begin{equation}
    h^+(\phi) = \begin{cases} 
    \phi^6 - a\phi^4 + (b - \frac{\omega^2}{2\lambda})\phi^2 & \text{if } \phi \geq 0, \\
    0 & \text{if } \phi < 0.
    \end{cases}
\end{equation}
Let $I_\omega^+$ denote the action functional associated with this truncated potential. It is straightforward to verify that $I_\omega^+$ satisfies the Palais-Smale condition and possesses the same Mountain Pass geometry as the original functional. Consequently, there exist a ground state $\tilde{\phi}$ (minimizer) and an excited state $\phi_{MP}$ (saddle point) for $I_\omega^+$.

To prove non-negativity, let $\phi$ be a critical point of $I_\omega^+$ and consider the test function $\phi^- = \min(0, \phi)$. Testing the Euler-Lagrange equation with $\phi^-$ yields:
\begin{align}
    \langle (I_\omega^+)'(\phi), \phi^- \rangle &= \int_0^P \rho \left( \phi_\rho \phi^-_\rho + \frac{N^2}{\rho^2}\phi\phi^- - \omega^2 \phi \phi^- \right) d\rho + \lambda \int_0^P \rho h^+(\phi) \phi^- d\rho.
\end{align}
Using the properties $\phi \phi^- = (\phi^-)^2$, $\phi_\rho \phi^-_\rho = (\phi^-_\rho)^2$, and the fact that $h^+(\phi)\phi^- \equiv 0$ by definition, this reduces to:
\begin{align}
    \int_0^P \rho (\phi^-_\rho)^2 d\rho + \int_0^P \rho \left( \frac{N^2}{\rho^2} - \omega^2 \right) (\phi^-)^2 d\rho = 0.
\end{align}
Since $\frac{N^2}{\rho^2} > \frac{N^2}{P^2}$ and we consider $\omega^2$ within the spectral gap, the term $(\frac{N^2}{\rho^2} - \omega^2)$ is bounded from below. The coercivity of the differential operator implies that $||\phi^-|| = 0$, and thus $\phi(\rho) \geq 0$ almost everywhere. 

Finally, by the Strong Maximum Principle applied to the operator $\mathcal{L} = -\Delta + (\frac{N^2}{\rho^2} - \omega^2)$, any non-trivial non-negative solution must be strictly positive in $(0, P)$. Since $\phi(\rho) > 0$, we have $h^+(\phi) \equiv h(\phi)$, and thus these positive critical points are genuine solutions to the original problem \eqref{DE}.
	
	\section{Proof of Theorem 1.3}
	\setcounter{equation}{0}
	Consider the action functional $I:W_0^{1,2}(0,P)\rightarrow\mathbb{R}$
	\begin{align}
		I(\phi)=\int_0^{P}\left\{\dfrac{1}{2}\left(\rho\phi_{\rho}^2+\dfrac{N^2}{\rho}\phi^2\right)+\lambda\rho(\phi^6-a\phi^4+b\phi^2)\right\}d\rho\label{action}
	\end{align}
	and the reduced norm functional $\tilde{Q}:W_0^{1,2}(0,P)\rightarrow\mathbb{R}$,
	\begin{align}
		\tilde{Q}(\phi)=4\pi\int_0^{P}\rho\phi^2d\rho.\label{constraint}
	\end{align}
	The field equation \eqref{DE} may be treated as a nonlinear eigenvalue problem such that $\omega^2$ appears as the Lagrange multiplier of the constrained minimization problem
	\begin{align}
		I_0(\tilde{Q}_0)=\inf\{I(\phi) : \phi \in H \text{ and } \tilde{Q}(\phi)=\tilde{Q}_0\}, \label{minimization}
	\end{align}
	where $\tilde{Q}_0 > 0$ is a fixed, prescribed value for the functional.
	
	Using the inequality $\phi^6-a\phi^4+b\phi^2\geq\phi^2(b-a^2/4)$, valid for all $\phi\in\mathbb{R}$, and the constraint $\tilde{Q}(\phi)=\tilde{Q}_0$, we get the coercive bound
	\begin{align}
		I(\phi)\geq \dfrac{1}{2}\int_0^{P}\left(\rho\phi_{\rho}^2+\dfrac{N^2}{\rho}\phi^2\right)d\rho+\dfrac{\lambda }{4\pi}\left(b-\dfrac{a^2}{4}\right)\tilde{Q}_0.\label{coercive2}
	\end{align}
	Since $b > a^2/4$, the functional is bounded below and the constrained minimization problem \eqref{minimization} is well defined. Let $\{\phi_j\}_{j=1}^{\infty}$ be a minimizing sequence of \eqref{minimization}. From the coercive inequality \eqref{coercive2}, there is a constant $C>0$ depending on $a$, $b$, $\lambda$, and $I_0$, but independent of $j$, such that
	\begin{align}
		\int_0^{P}\left(\rho\phi_{j,\rho}^2+\dfrac{N^2}{\rho}\phi_j^2\right)d\rho\leq C,\label{Cbnd}
	\end{align}
	with $\phi_{j,\rho}=\frac{d}{d\rho}\phi_j$. Since both functionals $I$ and $\tilde{Q}$ are even, and recalling that the distributional derivative satisfies $|\frac{d}{d\rho}|\phi(\rho)||\leq|\frac{d}{d\rho}\phi(\rho)|$ for any real-valued function defined on $W^{1,2}(0,P)$ \cite{LiebLoss}, we get $I(|\phi_j|)\leq I(\phi_j)$ and $\tilde{Q}(|\phi_j|)=\tilde{Q}(\phi_j)$. Hence, we may construct a new sequence, also labeled as $\{\phi_j\}_{j=1}^{\infty}$, of non-negative valued functions.
	
	We treat each $\phi_j$ as a radially symmetric function defined over $D_P=\{(x,y)\in\mathbb{R}^2:x^2+y^2\leq P^2\}$ vanishing on the boundary. Under the radially symmetric reduced norm $||\cdot||:W^{1,2}_0(D_P)\rightarrow\mathbb{R}$ given by
	\begin{align}
		||\phi||^2=\int_0^P\left(\rho\phi_{\rho}^2+\rho\phi^2\right)d\rho,
	\end{align}
	and the inequality 
    \begin{align}
    \int_0^P\rho\phi^2d\rho\leq P^2\int_0^P\frac{\phi^2}{\rho}d\rho,
	\end{align}
	we deduce that $\{\phi_j\}_{j=1}^{\infty}$ is bounded in $W^{1,2}_0(D_P)$. Without loss of generality, we suppose that $\phi_j\rightharpoonup \phi_0$ weakly in $W^{1,2}_0(D_P)$ as $j\rightarrow\infty$.
	The compact embedding $W^{1,2}(D_P)\subset\subset L^p(D_P)$ $(p\geq 1)$ gives the strong convergence of $\phi_j\rightarrow\phi_0$ in $L^p(D_P)$ as $j\rightarrow\infty$. Furthermore, for any $\epsilon \in (0, P)$, the sequence is bounded in $W^{1,2}(\epsilon, P)$, which embeds compactly into $\mathcal{C}[\epsilon, P]$. Thus, $\phi_j \to \phi_0$ uniformly on compact subsets away from the origin.
	
	To verify the constraint $\tilde{Q}(\phi_0) = \tilde{Q}_0$, we note that strong convergence in $L^2(D_P)$ implies
	\begin{align}
		\tilde{Q}(\phi_0) = 4\pi \int_0^P \rho \phi_0^2 d\rho = \lim_{j\to\infty} 4\pi \int_0^P \rho \phi_j^2 d\rho = \tilde{Q}_0.
	\end{align}
	The weak lower semicontinuity of the norm and the continuity of the potential terms under strong $L^p$ convergence imply
	\begin{align}
		I(\phi_0) \leq \liminf_{j\to\infty} I(\phi_j) = I_0.
	\end{align}
	Thus, $\phi_0$ is a minimizer. By the Lagrange Multiplier Theorem, there exists $\omega^2 \in \mathbb{R}$ such that $\phi_0$ satisfies the weak form of \eqref{DE}. Standard elliptic regularity theory implies that $\phi_0$ is a classical solution $\phi_0 \in \mathcal{C}^2(0, P) \cap \mathcal{C}[0, P]$.
	
	To prove positivity, suppose there exists $\rho_0 \in (0, P)$ such that $\phi_0(\rho_0) = 0$. Since $\phi_0 \geq 0$, this $\rho_0$ is a global minimum, implying $\phi_{0,\rho}(\rho_0) = 0$. By the uniqueness theorem for ODEs, this initial condition $\phi_0(\rho_0) = \phi'_{0}(\rho_0) = 0$ implies $\phi_0(\rho) \equiv 0$ for all $\rho$, which contradicts the constraint $\tilde{Q}(\phi_0) = \tilde{Q}_0 > 0$. Therefore, $\phi_0(\rho) > 0$ for all $\rho \in (0, P)$.
	
	Finally, we establish the lower bound on the reduced norm. Multiply \eqref{DE} by $\rho\phi(\rho)$ and integrate by parts on $(0,P)$. Using the equation satisfied by the minimizer $(\phi_0, \omega)$:
	\begin{align}
		0 &= \int_0^{P} \left( -\rho \phi_{\rho}^2 - \frac{N^2}{\rho}\phi^2 - \lambda\rho(6\phi^6 - 4a\phi^4 + 2b\phi^2) + \omega^2\rho\phi^2 \right) d\rho \nonumber \\
		&\leq -\int_0^{P} \rho \phi_{\rho}^2 d\rho - N^2\int_0^P\dfrac{\phi^2}{\rho}d\rho + 4a\lambda\int_0^P\rho\phi^4d\rho - (2\lambda b-\omega^2)\int_0^P\rho\phi^2d\rho.\label{rearrange1}
	\end{align}
	Using the Cauchy-Schwarz inequality and $\phi(0)=0$:
	\begin{align}
		\phi^2(\rho) = \int_0^{\rho} 2\phi(s)\phi_{s}(s)ds \leq 2\left(\int_0^P s\phi_{s}^2 ds\right)^{1/2}\left(\int_0^P \dfrac{\phi^2}{s} ds\right)^{1/2}.
	\end{align}
	Multiplying by $\rho \phi^2(\rho)$ and integrating, we obtain:
	\begin{align}
		\int_0^{P} \rho\phi^4 d\rho \leq \frac{\tilde{Q}_0}{2\pi} \left(\int_0^{P}\rho \phi_{\rho}^2 d\rho\right)^{1/2}\left(\int_0^{P}\dfrac{\phi^2}{\rho}d\rho\right)^{1/2}.
	\end{align}
	Applying Young's inequality $xy \leq \epsilon x^2 + \frac{y^2}{4\epsilon}$ with $x = (\int \rho \phi_\rho^2)^{1/2}$ and $y = \frac{\tilde{Q}_0}{2\pi} (\int \phi^2/\rho)^{1/2}$:
	\begin{align}
		\int_0^{P} \rho\phi^4 d\rho \leq \epsilon \int_0^{P}\rho \phi_{\rho}^2 d\rho + \frac{1}{4\epsilon} \left(\frac{\tilde{Q}_0^2}{4\pi^2}\int_0^{P}\dfrac{\phi^2}{\rho}d\rho\right).\label{eq411}
	\end{align}
	Substituting \eqref{eq411} into \eqref{rearrange1}, the positive gradient term has coefficient $(4a\lambda\epsilon - 1)$:
	\begin{align}
		0 \leq (4a\lambda\epsilon-1)\int_0^{P}\rho \phi^2_{\rho}d\rho + \left(\frac{4a\lambda \tilde{Q}_0^2}{16\epsilon \pi^2} - N^2\right)\int_0^{P}\dfrac{\phi^2}{\rho}d\rho - (2\lambda b-\omega^2) \frac{\tilde{Q}_0}{4\pi}.
	\end{align}
	We choose $\epsilon = \frac{1}{4a\lambda}$ to eliminate the gradient term. The coefficient of the integral $\int (\phi^2/\rho)$ becomes:
	\begin{align}
		\frac{4a\lambda \tilde{Q}_0^2}{16 \left(\frac{1}{4a\lambda}\right) \pi^2} - N^2 = \frac{16 a^2 \lambda^2 \tilde{Q}_0^2}{16 \pi^2} - N^2 = \frac{a^2 \lambda^2 \tilde{Q}_0^2}{\pi^2} - N^2.
	\end{align}
	Thus, the inequality becomes:
	\begin{align}
		0 \leq \left(\frac{a^2\lambda^2 \tilde{Q}_0^2}{\pi^2} - N^2\right) \int_0^{P}\dfrac{\phi^2}{\rho}d\rho - (2\lambda b-\omega^2) \frac{\tilde{Q}_0}{4\pi}.
	\end{align}
	For a solution to exist with $\omega^2 < 2\lambda b$, the last term is strictly negative. Therefore, for the inequality to hold, the first term must be strictly positive:
	\begin{align}
		\frac{a^2\lambda^2 \tilde{Q}_0^2}{\pi^2} - N^2 > 0 \implies \tilde{Q}_0 > \frac{\pi |N|}{a\lambda}.
	\end{align}
	
	
\section{Numerical Results: Profiles of fundamental $Q$-vortices}
\setcounter{equation}{0}
The constrained minimization method of Section 4 enables the spectral-Galerkin computation of fundamental $Q$-vortex profiles for a prescribed reduced norm $\tilde{Q}_0$. Identical to the numerical work of Volkov and W\"ohnert \cite{VW}, we adopt the parameter set: $\lambda=1$, $a=2$, and $b=1.1$.

These parameters impose strict theoretical constraints on the existence and geometry of the solutions, which serve as benchmarks for our numerical results:
\begin{enumerate}
    \item \textit{Amplitude Bound:} Theorem 1.1 establishes a uniform upper bound on the field amplitude, predicting saturation at:
    \begin{equation}
        \phi^2(\rho) < \frac{2a}{3} \implies |\phi(\rho)| < \sqrt{\frac{4}{3}} \approx 1.1547.
    \end{equation}
    \item \textit{Frequency Existence Range:} Theorem 1.2 guarantees the existence of ground and excited state solutions for angular frequencies within the range:
    \begin{equation}
        \omega_{min}^2 := 2\lambda\left(b-\frac{a^2}{4}\right) < \omega^2 < 2\lambda b =: \omega_{max}^2.
    \end{equation}
    For our chosen parameters, these critical bounds are $\omega_{min}^2 = 0.2$ and $\omega_{max}^2 = 2.2$.
    \item \textit{Prescribed Norm Threshold:} Theorem 1.3 predicts that for solutions to exist below the upper frequency bound ($\omega^2 < 2.2$), the prescribed reduced norm must satisfy the lower bound:
    \begin{equation}
        \tilde{Q}_0 > \frac{\pi |N|}{a\lambda} = \frac{\pi |N|}{2} \approx 1.57 |N|.
    \end{equation}
\end{enumerate}

To determine the appropriate domain size for the numerical computation, we refer to the sufficient bound $P^*$ established in Lemma 3.2. For the chosen parameters, we calculate $P^* \approx 12.5$. To ensure our domain is sufficiently large to capture the exponential tail of the solution and minimize boundary effects, we conservatively select $P=20$ for all subsequent computations as in \cite{VW}.

Let $V_m$ be the finite-dimensional subspace spanned by the first $m$ orthonormal functions $\{\psi_j\}_{j=1}^m \subset W^{1,2}_0(0, P)$, constructed from the sine basis via a modified Gram-Schmidt process. We define the weighted inner product:
\begin{equation}
    (\psi, \tilde{\psi}) = 4\pi \int_0^P \rho \psi(\rho) \tilde{\psi}(\rho) \, d\rho,
\end{equation}
such that $(\psi_i, \psi_j) = \delta_{ij}$. This choice is motivated by the reduced norm functional $\tilde{Q}(\phi)$, allowing the constraint to be expressed simply as the sum of squared coefficients.

We approximate the solution as $\phi(\rho) \approx \sum_{j=1}^{m} a_j \psi_j(\rho)$, where $a \in \mathbb{R}^m$. Substituting this ansatz into the action functional $I(\phi)$ and exploiting the orthonormality of the basis, the potential term $\int \rho b \phi^2$ reduces to a constant proportional to $\tilde{Q}_0$. The infinite-dimensional problem reduces to minimizing the function $F: \mathbb{R}^m \to \mathbb{R}$:
\begin{align}
    F(a) &= \frac{1}{2} \sum_{i,j=1}^m a_i a_j \left( K_{ij} + N^2 C_{ij} \right) + \frac{\lambda b}{8\pi}\tilde{Q}_0 + \lambda \int_0^P \rho \left( \phi^6(\rho) - a\phi^4(\rho) \right) d\rho,
    \label{DiscreteFunctional}
\end{align}
subject to the constraint $\sum_{j=1}^m a_j^2 = \tilde{Q}_0$. Here, $K_{ij}$ and $C_{ij}$ are the pre-computed ``stiffness" and ``centrifugal" matrices, respectively:
\begin{equation}
    K_{ij} = \int_0^P \rho \psi_i'(\rho) \psi_j'(\rho) \, d\rho, \qquad C_{ij} = \int_0^P \frac{1}{\rho} \psi_i(\rho) \psi_j(\rho) \, d\rho.
\end{equation}
The minimization of \eqref{DiscreteFunctional} is performed using Python's \texttt{SciPy} optimization library. Once the minimizer $\phi$ is found, the nonlinear frequency shift $\omega^2$ (the Lagrange multiplier) is recovered by projecting the Euler-Lagrange equation onto the solution itself:
\begin{equation}
    \omega^2 = \frac{4\pi}{\tilde{Q}_0} \left( \int_0^P \rho \phi_\rho^2 \, d\rho + N^2 \int_0^P \frac{\phi^2}{\rho} \, d\rho + \int_0^P \rho \phi U'(\phi) \, d\rho \right),
    \label{OmegaEq}
\end{equation}
where $U'(\phi) = \lambda(6\phi^5 - 4a\phi^3 + 2b\phi)$.

To quantify the accuracy of our numerical approximation, we define the \textit{Residual Error} (RE) as the $L^2$ norm of the residual of the differential equation \eqref{DE}, normalized by the domain size,
\begin{equation}
	RE = \frac{1}{P} \left( \int_0^P \left( \phi_{\rho\rho} + \frac{1}{\rho}\phi_{\rho} - \frac{N^2}{\rho^2}\phi + \omega^2\phi - \lambda(6\phi^5 - 4a\phi^3 + 2b\phi) \right)^2 d\rho \right)^{1/2}.
	\label{ResidualError}
\end{equation}
Figure \ref{fig:fig1} displays the computed profiles for varying prescribed norm $\tilde{Q}_0$ with fixed $N=1$. The profiles exhibit a characteristic transition from Gaussian-like shapes at low norm to flat-top, saturated structures at high norm. This saturation behavior is predicted by the potential structure, and we observe that the maximum amplitude remains strictly bounded below the theoretical limit $\phi_{max} \approx 1.1547$ (dashed line), confirming the uniform bound derived in Theorem 1.1.

\begin{figure}
    \centering
    \includegraphics[width=0.8\textwidth]{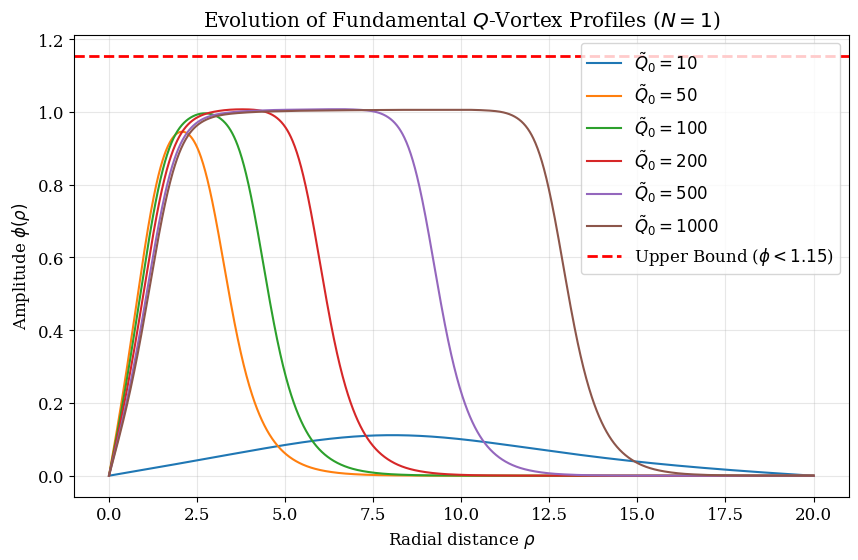}
    \caption{Evolution of the fundamental $Q$-vortex profile ($N=1$) with increasing prescribed norm $\tilde{Q}_0$. The amplitude saturates and flattens as $\tilde{Q}_0$ gets large, strictly respecting the theoretical upper bound $\phi < \sqrt{2a/3}$ (red dashed line).}
    \label{fig:fig1}
\end{figure}

To better visualize the spatial geometry and the saturation effect, Figure \ref{fig:compare3D} presents a side-by-side 3D reconstruction of the vortex for $\tilde{Q}_0=100$ and $\tilde{Q}_0=500$. For the moderate norm $\tilde{Q}_0=100$, the profile retains a smooth, rounded peak. In contrast, for the large norm $\tilde{Q}_0=500$, the profile clearly exhibits a ``flat-top" structure where the amplitude hits the theoretical ceiling and the vortex expands radially to accommodate the increased norm, rather than growing in height.

\begin{figure}[H]
    \centering
    \includegraphics[width=\textwidth]{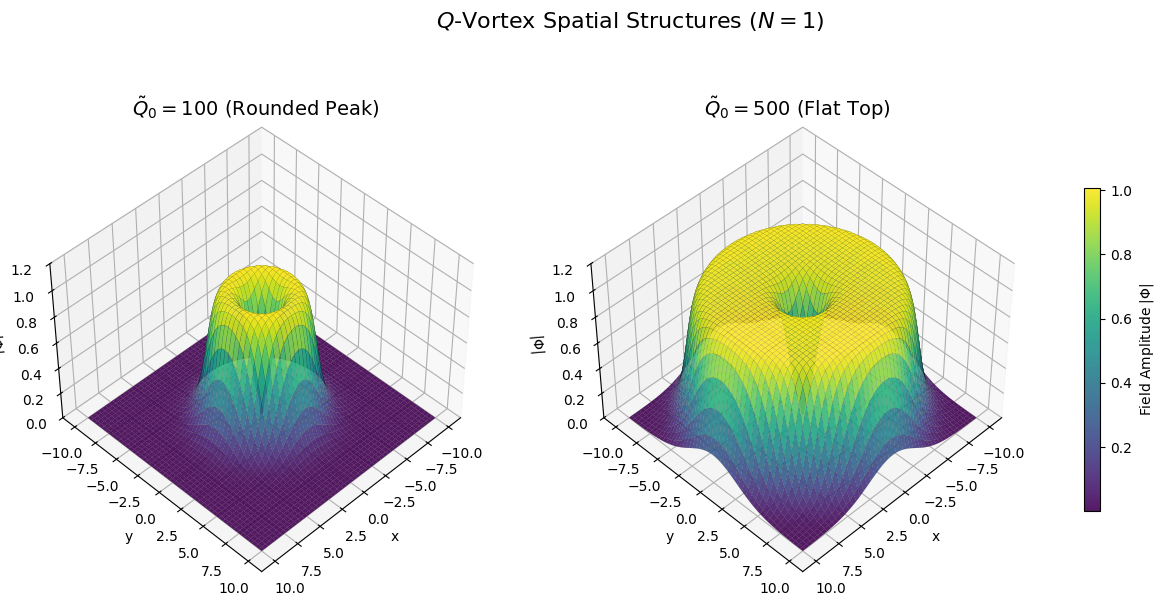}
    \caption{Comparison of the spatial structure of fundamental $Q$-vortices ($N=1$) for moderate ($\tilde{Q}_0=100$) and large ($\tilde{Q}_0=500$) prescribed norm. The left panel shows a rounded peak, while the right panel illustrates the ``flat-top" saturation predicted by Theorem 1.1, where the amplitude is capped and the soliton widens.}
    \label{fig:compare3D}
\end{figure}

Table \ref{table:results} presents the computed parameters for selected values of the prescribed norm, verifying that the residual errors remain small across the entire regime.

\begin{table}[H]
    \centering
    \caption{Computed parameters for fundamental $Q$-vortices ($N=1$).}
    \label{table:results}
    \begin{tabular}{|c|c|c|c|}
        \hline
        Norm $\tilde{Q}_0$ & Frequency $\omega^2$ & Max Amplitude $\phi_{max}$ & Residual Error (RE) \\
        \hline
        10.0   & 2.1755 & 0.1115 & $4.17 \times 10^{-4}$ \\
        50.0   & 0.5663 & 0.9458 & $3.55 \times 10^{-2}$ \\
        100.0  & 0.4287 & 0.9963 & $4.47 \times 10^{-2}$ \\
        200.0  & 0.3517 & 1.0073 & $3.35 \times 10^{-2}$ \\
        500.0  & 0.2904 & 1.0077 & $4.07 \times 10^{-2}$ \\
        1000.0 & 0.2618 & 1.0062 & $5.47 \times 10^{-2}$ \\
        \hline
    \end{tabular}
\end{table}

Figure \ref{fig:fig2} illustrates the nonlinear dispersion relation of the vortex, specifically the dependence of the eigenfrequency $\omega^2$ on the prescribed norm $\tilde{Q}_0$. Physically, $\tilde{Q}_0$ acts as the independent control parameter determining the soliton's ``mass" or geometric size. As we increase the prescribed value $\tilde{Q}_0$, the system responds by lowering its eigenfrequency $\omega^2$, decreasing monotonically from the linear limit ($\omega^2_{max}=2.2$) toward the critical lower bound defined in (5.2). This asymptotic behavior implies that stable $Q$-vortices can accommodate an arbitrarily large reduced norm, provided the frequency remains strictly above the cutoff $\omega_{min}^2 = 0.2$.

\begin{figure}[H]
    \centering
    \includegraphics[width=0.8\textwidth]{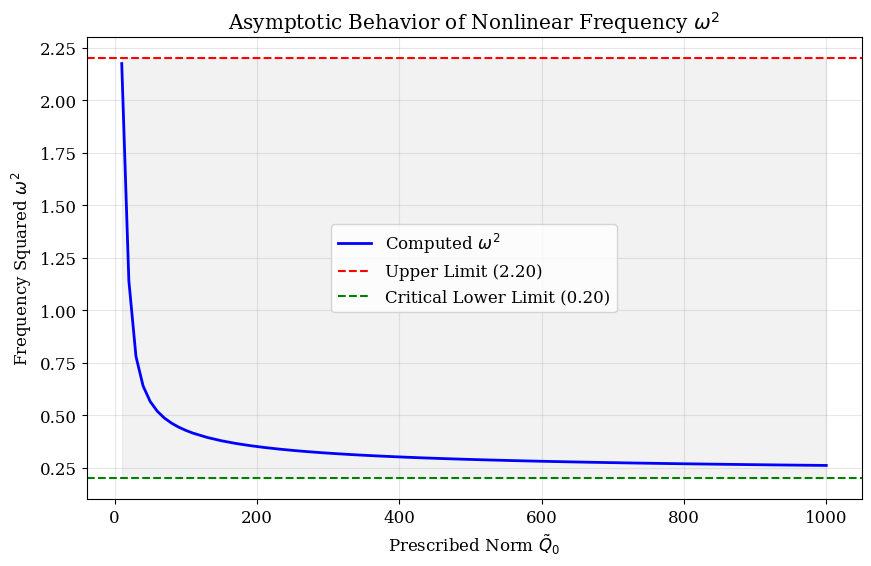}
    \caption{Asymptotic behavior of the nonlinear frequency shift $\omega^2$ as a function of prescribed norm $\tilde{Q}_0$. The frequency monotonically approaches the critical lower bound $\omega^2_{min}=0.2$ for arbitrarily large prescribed norm.}
    \label{fig:fig2}
\end{figure}

Finally, we investigate the effect of the rotational quantum number $N$ on the soliton properties. We fix the prescribed norm at $\tilde{Q}_0 = 100$ and vary $N$ from 1 to 5. As shown in Figure \ref{fig:fig3} and Table \ref{table:N_effects}, increasing $N$ pushes the vortex core outward. This effect is driven by the \textit{centrifugal barrier} term $N^2/\rho^2$ in the energy functional, which acts as a repulsive potential near the origin. To minimize energy, the field amplitude is forced to vanish at $\rho=0$ and the peak density is shifted to larger radii. Consequently, the maximum amplitude $\phi_{max}$ decreases significantly, while the eigenfrequency $\omega^2$ increases, moving closer to the upper spectral bound.

\begin{figure}[H]
    \centering
    \includegraphics[width=0.8\textwidth]{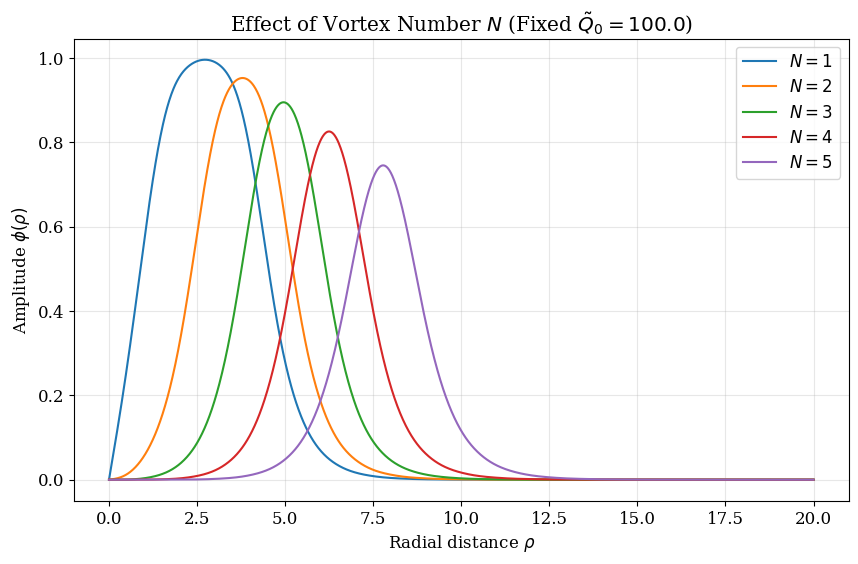}
    \caption{Effect of increasing vortex number $N$ on the soliton profile for fixed prescribed norm $\tilde{Q}_0=100$. The centrifugal barrier $N^2/\rho^2$ forces the wavefunction away from the origin, reducing the peak amplitude and increasing the effective radius.}
    \label{fig:fig3}
\end{figure}

\begin{table}[H]
    \centering
    \caption{Effect of vortex number $N$ on solution properties for fixed prescribed norm $\tilde{Q}_0 = 100$.}
    \label{table:N_effects}
    \begin{tabular}{|c|c|c|c|}
        \hline
        Vortex Number $N$ & Frequency $\omega^2$ & Max Amplitude $\phi_{max}$ & Residual Error (RE) \\
        \hline
        1 & 0.4287 & 0.9963 & $4.47 \times 10^{-2}$ \\
        2 & 0.5351 & 0.9530 & $8.10 \times 10^{-2}$ \\
        3 & 0.6657 & 0.8954 & $1.77 \times 10^{-2}$ \\
        4 & 0.8239 & 0.8261 & $3.33 \times 10^{-2}$ \\
        5 & 1.0145 & 0.7457 & $2.59 \times 10^{-2}$ \\
        \hline
    \end{tabular}
\end{table}

To confirm that these amplitude profiles correspond to topologically distinct states, we visualize the full complex field $\Phi = \phi(\rho)e^{iN\theta}$. Figure \ref{fig:phase3D} maps the phase angle $\arg(\Phi)$ onto the 3D surface of the amplitude for the $N=1$ and $N=2$ cases. The color winding visibly demonstrates the topological charge: for $N=1$, the phase cycles through the spectrum ($-\pi$ to $\pi$) exactly once around the ring, while for $N=2$, it cycles twice. This visualization intuitively explains the ``hole" at the core: the phase singularity at the origin forces the amplitude to vanish to maintain smoothness, and higher winding numbers ($N=2$) require a wider core region to accommodate the faster phase variation.

\begin{figure}[H]
    \centering
    \includegraphics[width=\textwidth]{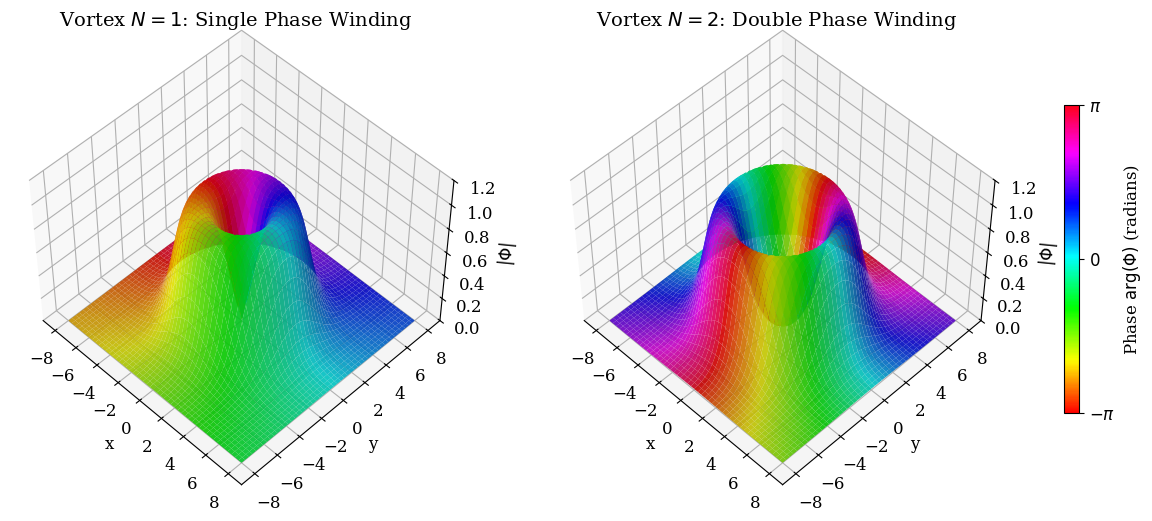}
    \caption{Phase-colored 3D visualization of the $Q$-vortex for $N=1$ (left) and $N=2$ (right). The height represents the field amplitude $|\Phi|$, while the color represents the phase angle $\arg(\Phi)$. The number of full color cycles around the ring visually confirms the topological winding number $N$. Note the wider core for $N=2$, consistent with the stronger centrifugal barrier.}
    \label{fig:phase3D}
\end{figure}

	
	\clearpage
    \begin{acknowledgements}
    We thank the anonymous referee for the careful reading of this manuscript and helpful feedback. This work does not have any conflicts of interest. There are no funders to report for this submission.  
\end{acknowledgements}
	\small{
	
}
	

\begin{thebibliography}{99}

		\bibitem{Almumin2024}
		Y. Almumin, J. Heeck, A. Rajaraman, C. B. Verhaaren, \textit{Slowly rotating Q-balls}, Eur. Phys. J. C \textbf{84}, 364 (2024).

		\bibitem{Ansari2024}
		A. Ansari, L. S. Bhandari, and A. M. Thalapillil, \textit{Q-balls in the sky}, Phys. Rev. D \textbf{109}, 023003 (2024).

		\bibitem{Benci3}
		V. Benci and D. Fortunato, \textit{Variational methods in nonlinear field equations: solitary waves, hylomorphic solitons, and vortices}, Springer monographs in mathematics, Springer, Cham (2014).
		
		\bibitem{Benci4}
		V. Benci and D. Fortunato, \textit{Spinning Q-balls for the Klein-Gordon-Maxwell equations}, Commun.
		Math. Phys. 295, 639–668 (2010).
		
		\bibitem{Benci5}
		V. Benci and D. Fortunato, \textit{On the existence of stable charged Q-balls}, J. Math. Phys. 52, (2011).
		
		\bibitem{Birse}
		M. C. Birse, \textit{Nontopological solitons}, Nucl. Phys. A \textbf{543} 337-348 (1992).
		
		\bibitem{Brihaye2024}
		Y. Brihaye and F. Buisseret, \textit{Q-balls and charged Q-balls in a two-scalar field theory with generalized Henon-Heiles potential}, Phys. Rev. D \textbf{109}, 076029 (2024).

		\bibitem{BHZ}
		Y. Brihaye, B. Hartmann, and W. Zakrzewski, \textit{Spinning solitons of a modified nonlinear Schr\"odinger equation}, Phys. Rev. D \textbf{69}, 087701 (2004).
		
		\bibitem{BCHL}
		Y. Brihaye, A. Cisterna, B. Hartmann, and G. Luchini, \textit{From topological to nontopological solitons: kinks, domain walls, and Q-balls in a scalar field model with a nontrivial vacuum manifold}, Phys. Rev. D \textbf{92}, 124061 (2015).
		
		\bibitem{Coleman}
		S. Coleman, \textit{Q-balls}, Nucl. Phys. B \textbf{262}, 263-283, (1985).

		\bibitem{Derrick1964}
		G. H. Derrick, \textit{Comments on nonlinear wave equations as models for elementary particles}, J. Math. Phys. \textbf{5}, 1252 (1964).
		
		\bibitem{DW}
		S. Dodelson, L. Widrow, \textit{Baryon symmetric baryogenesis}, Phys. Rev. Lett. \textbf{64}, 340–343, (1990).

		\bibitem{Hamada2024}
		Y. Hamada, K. Kawana, T. Kim, and P. Lu, \textit{Q-balls in the presence of attractive force}, J. High Energy Phys. \textbf{2024}, 242 (2024).
		
		\bibitem{Jabri}
		Y. Jabri, {\em The Mountain Pass Theorem, Variants, Generalizations and Some Applications}, Encyclopedia Math. Appl., Cambridge University Press, Cambridge, UK, 2003.
		
		\bibitem{KKK}
		Chanju Kim, Seyong Kim, and Yoonbai Kim, \textit{Global nontopological solitons}, Phys. Rev. D \textbf{47}, 5434-5443 (1993).

		\bibitem{GaugedCollisions2024}
		M. P. Kinach and M. W. Choptuik, \textit{Relativistic head-on collisions of U(1) gauged Q-balls}, Phys. Rev. D \textbf{110}, 015012 (2024).
		
		\bibitem{Kusenko}
		A. Kusenko, M. Shaposhnikov, \textit{Supersymmetric Q balls as dark matter}, Phys. Lett. B \textbf{418}
		46–54 (1998).
		
		\bibitem{LeePang}
		T. D. Lee and Y. Pang, \textit{Nontopological solitons}, Phys. Rep. \textbf{221}, 251-350 (1992).
		
		\bibitem{LiebLoss}
		E. H. Lieb and M. Loss, \emph{Analysis}, Amer. Math. Soc., Providence, 1997.


		\bibitem{Ivashkin2025}
        I. Ivashkin, E. Kim, E. Nugaev, and Y. Shnir, Towards spinning $U(1)$ gauged non-topological solitons in the model with Chern--Simons term, Phys. Lett. B \textbf{873}, 140229 (2026).

		\bibitem{Long2025}
		Y. Long, C. Ren, H. Chen, and H. Ge, \textit{Uncovering Hidden Spin of Scalar Fields with Higher-Order Derivative Lagrangian: On the Wave Spin in Drifted and Dissipative Fields}, Chin. Phys. Lett. \textbf{42}, 064301 (2025).

		\bibitem{Mandal2022}
		S. Mandal, \textit{Solitons in curved spacetime}, EPL (Europhysics Letters) \textbf{136}, 11001 (2022).

		\bibitem{Palais} R. S. Palais, \textit{The principle of symmetric criticality}, Commun. Math. Phys. \textbf{69}, 19–30 (1979).
		
		\bibitem{PR}
		P. H. Rabinowitz, \textit{Minimax methods in critical point theory with applications to differential equations}, American Mathematical Soc. No. 65, 1986.
		
		\bibitem{SuHan2025}
		G. Su and X. Han, \textit{Vortex Solutions for A Mixed Boundary-Value Problem in the Abelian-Higgs Model with A Neutral Scalar Field}, arXiv preprint arXiv:2511.06931 (2025).

		\bibitem{Tai2024}
		J.-S. B. Tai, \textit{Topological solitons in chiral liquid crystals}, Liq. Cryst. Today \textbf{32}, 45–62 (2024).

		\bibitem{VW} M. S. Volkov and E. W\"ohnert, \textit{Spinning $Q$-balls}, Phys. Rev. D \textbf{66}, 085003 (2002).
		
		\bibitem{Zhang2025}
		G.-D. Zhang, C.-H. Li, Q.-X. Xie, S.-Y. Zhou, \textit{Superradiance of Friedberg-Lee-Sirlin solitons}, Phys. Rev. D \textbf{111}, 103027 (2025).

        \bibitem{ChenSu2022}
        S. Chen and G. Su, Existence of vortices for Schr\"odinger equations with logarithmic and saturable nonlinearity, J. Math. Phys. 63, 101506 (2022).

        \bibitem{Guo2019}
        Q. Guo, D. Cao, and H. Li, Existence of optical vortices, Nonlinear Anal. Real World Appl. 50, 67 (2019).

        \bibitem{Greco2016}
        C. Greco, On the cubic and cubic-quintic optical vortices equations, J. Appl. Anal. 22, 95 (2016).

        \bibitem{Medina2017}
        L. Medina, On the existence of optical vortex solitons propagating in saturable nonlinear media, J. Math. Phys. 58, 011505 (2017).
        
        \bibitem{Medina2021}
        L. Medina, Localized optical vortex solitons in pair plasmas, J. Appl. Anal. 27, 1 (2021).

        \bibitem{Medina2023}
        L. Medina, Existence of Coupled Optical Vortex Solitons Propagating in a Quadratic Nonlinear Medium, Math. Meth. Appl. Sci. 46, 18547 (2023).

        \bibitem{YZ}
        Y. Yang and R. Zhang, \textit{Existence of optical vortices}, SIAM J. Math. Anal. \textbf{46} (2014) 484-498.
		
	\end{thebibliography}
\end{document}